\documentclass[11pt]{article}

\newif\ifblind
\blindfalse

\newif\ifdraft
\draftfalse

\usepackage{rpmacros}
\usepackage[dvipsnames,svgnames,x11names,hyperref]{xcolor}
\usepackage{stmaryrd}
\usepackage{amsthm}
\usepackage{xfrac}
\usepackage{mathpazo}
\usepackage{gitinfo}
\usepackage{enumitem}
\usepackage{fullpage}
\usepackage{thmtools}
\usepackage[T1]{fontenc} 
\usepackage[colorlinks]{hyperref}
\hypersetup{
  linkcolor=[rgb]{0,0,0.4},
  citecolor=[rgb]{0, 0.4, 0},
  urlcolor=[rgb]{0.6, 0, 0}
}
\usepackage{lipsum}
\usepackage{calrsfs,bbm}
\usepackage{tikz}
\usetikzlibrary{decorations.pathreplacing,backgrounds}
\usepackage{multirow}
\usepackage{setspace}

\usepackage{microtype}

\usepackage[linesnumbered,vlined,ruled]{algorithm2e}
\usepackage{mathtools}

\usepackage{amsthm}
\usepackage{thmtools,thm-restate}
\usepackage[nameinlink,capitalise]{cleveref}

\numberwithin{equation}{section}
\declaretheoremstyle[bodyfont=\it,qed=\qedsymbol]{noproofstyle}

\declaretheorem[name=Observation,numbered=no]{observation*}

\declaretheorem[numberlike=equation]{fact}

\declaretheorem[numberlike=equation]{theorem}

\declaretheorem[name=Theorem,numbered=no]{theorem*}

\declaretheorem[numberlike=equation]{lemma}
\declaretheorem[name=Lemma,numbered=no]{lemma*}
\declaretheorem[numberlike=equation,style=noproofstyle,name=Lemma]{lemmawp}

\declaretheorem[numberlike=equation]{corollary}
\declaretheorem[name=Corollary,numbered=no]{corollary*}
\declaretheorem[numberlike=equation,style=noproofstyle,name=Corollary]{corollarywp}

\declaretheorem[name=Proposition,numbered=no]{proposition*}

\declaretheorem[numberlike=equation]{claim}
\declaretheorem[name=Claim,numbered=no]{claim*}

\declaretheorem[name=Conjecture,numbered=no]{conjecture*}

\declaretheorem[name=Question,numbered=no]{question*}

\declaretheoremstyle[bodyfont=\it,qed=$\lozenge$]{defstyle} 

\declaretheorem[numberlike=equation,style=defstyle]{definition}
\declaretheorem[unnumbered,name=Definition,style=defstyle]{definition*}

\declaretheorem[unnumbered,name=Example,style=defstyle]{example*}

\declaretheorem[unnumbered,name=Notation=defstyle]{notation*}

\declaretheorem[unnumbered,name=Construction,style=defstyle]{construction*}

\declaretheorem[unnumbered,name=Remark,style=defstyle]{remark*}

\renewcommand{\phi}{\varphi}
\renewcommand{\epsilon}{\varepsilon}

\newcommand{\size}{\operatorname{size}}
\newcommand{\depth}{\operatorname{depth}}

\newcommand{\SP}{\Sigma\Pi}
\newcommand{\SPsize}[1]{(\SP)^k\operatorname{-size}}

\newcommand{\hasse}[2]{\operatorname{D}^{(#1)}_{#2}}

\usepackage{nth}
\usepackage{intcalc}
\usepackage{etoolbox}
\usepackage{xstring}

\usepackage{ifpdf}
\ifpdf
\else
\usepackage[quadpoints=false]{hypdvips}
\fi

\newcommand{\ECCCurl}[2]{http://eccc.hpi-web.de/report/\ifnumcomp{#1}{>}{93}{19}{20}#1/#2/}

\newcommand{\shortECCC}[2]{\texttt{\href{http://eccc.hpi-web.de/report/\ifnumcomp{#1}{>}{93}{19}{20}#1/#2/}{eccc:TR#1-#2}}}

\newcommand{\parseECCC}[1]{
\StrSubstitute{#1}{TR}{}[\tmpstring]%
\IfSubStr{\tmpstring}{/}{ 
\StrBefore{\tmpstring}{/}[\ecccyear]%
\StrBehind{\tmpstring}{/}[\ecccreport]%
}{
\StrBefore{\tmpstring}{-}[\ecccyear]%
\StrBehind{\tmpstring}{-}[\ecccreport]%
}%
\shortECCC{\ecccyear}{\ecccreport}}

\newcommand{\homog}{\operatorname{Hom}}
\newcommand*\samethanks[1][\value{footnote}]{\footnotemark[#1]}

\newcounter{todo}

\ifdraft
\newcommand{\RPnote}[1]{\refstepcounter{todo}\textcolor{WildStrawberry}{\guillemotleft RP: #1 \guillemotright\addcontentsline{tod}{subsection}{[RP]~#1}}}
\newcommand{\MKnote}[1]{\refstepcounter{todo}\textcolor{BlueGreen}{\guillemotleft Mrinal: #1 \guillemotright\addcontentsline{tod}{subsection}{[MK]~#1}}}
\newcommand{\VRnote}[1]{\refstepcounter{todo}\textcolor{Blue}{\guillemotleft VR: #1 \guillemotright\addcontentsline{tod}{subsection}{[VR]~#1}}}
\newcommand{\SBnote}[1]{\refstepcounter{todo}\textcolor{OliveGreen}{\guillemotleft SB: #1 \guillemotright\addcontentsline{tod}{subsection}{[SB]~#1}}}
\newcommand{\SSnote}[1]{\refstepcounter{todo}\textcolor{BrickRed}{\guillemotleft SS: #1 \guillemotright\addcontentsline{tod}{subsection}{[SS]~#1}}}

\newcommand{\gitinfonotecolour}{Gray}
\newcommand{\easteregg}{}
\else
\newcommand{\RPnote}[1]{}
\newcommand{\MKnote}[1]{}
\newcommand{\VRnote}[1]{}
\newcommand{\SBnote}[1]{}
\newcommand{\SSnote}[1]{}

\newcommand{\gitinfonotecolour}{white}
\newcommand{\easteregg}{`Yes, but what about characteristic $p$?'}
\fi

\newcommand{\ignore}[1]{}
\newcommand{\gitinfonote}{git info:~\gitAbbrevHash\;,\;(\gitAuthorIsoDate)\; \;\gitVtag}

\newcommand{\Res}[3]{\ensuremath{\operatorname{Res}_{#1}(#2,#3)}} 

\newcommand{\K}{\ensuremath{\mathbb{K}}}

\newcommand{\coeff}[2]{[#1]\inbrace{#2}}

\newcommand{\diag}{\mathcal{D}}

\newcommand{\esym}{\mathrm{Esym}}
\newcommand{\Filter}{\operatorname{Filter}}

\makeatletter

\newcommand\listtodoname{List of todos}
\newcommand\listoftodos{%
  \section*{\listtodoname}\@starttoc{tod}}
\makeatother

\allowdisplaybreaks
\title{Constant-depth circuits for polynomial GCD\\over any characteristic} 

\ifblind
\else
\author{
     {Somnath Bhattacharjee \thanks{University of Toronto, Canada. Email: \texttt{somnath.bhattacharjee@mail.utoronto.ca}. Research partially supported by an NSERC Discovery Grant.}}
     \and
     {Mrinal Kumar \thanks{Tata Institute of Fundamental Research, Mumbai, India. Email: \texttt{\{mrinal, shanthanu.rai, varun.ramanathan, ramprasad\}@tifr.res.in}.  Research supported by the Department of Atomic Energy, Government of India, under project number RTI400112, and in part by Google and SERB Research Grants. }}
     \and
     {Shanthanu S. Rai{\samethanks[2]}}
     \and
     {Varun Ramanathan{\samethanks[2]}}
     \and
     {Ramprasad Saptharishi{\samethanks[2]}}
     \and
     {Shubhangi Saraf \thanks{University of Toronto, Canada. Email: \texttt{shubhangi.saraf@utoronto.ca}. Research partially supported by the McLean Award and an NSERC Discovery Grant.}}
}
\fi
\date{}

\begin{document}

\maketitle

\begin{abstract}
We show that the GCD of two univariate polynomials can be computed by (piece-wise) algebraic circuits of constant depth and polynomial size over any sufficiently large field, regardless of the characteristic. This extends a recent result of Andrews \& Wigderson  who showed such an upper bound over fields of zero or large characteristic. 

Our proofs are based on a recent work of Bhattacharjee, Kumar, Rai, Ramanathan, Saptharishi \& Saraf that shows  closure of constant depth algebraic circuits under factorization. On our way to the proof, we show that any $n$-variate symmetric polynomial $P$ that has a small constant depth algebraic circuit can be written as the composition of a small constant depth algebraic circuit with elementary symmetric polynomials. This statement is a constant depth version of a result of Bl\"{a}ser \& Jindal, who showed this for algebraic circuits of unbounded depth. As an application of our techniques, we also strengthen the closure results for factors of constant-depth circuits in the work of Bhattacharjee et al. over fields for small characteristic. 
\end{abstract}


\newpage 

\tableofcontents

\section{Introduction}
A popular theme in the research on algebraic computation is that of parallelization. A typical result of this flavor shows that a seemingly sequential computational task can also be performed by algebraic circuits of surprisingly small depth. One of the earliest results of this nature is a result of Csanky \cite{C76} which shows that the $n\times n$ symbolic determinant polynomial can be efficiently computed by algebraic circuits of depth $O(\log^2 n)$. Subsequently, a significant generalization of this result by Valiant, Skyum, Berkowitz and Rackoff \cite{VSBR83} showed that any algebraic circuit of size and degree $\poly(n)$ can be converted to a circuit of depth $O(\log^2 n)$ with a $\poly(n)$ blow up in size. The decades of 1980s saw many such surprising parallelizability results for algebraic problems, and we refer an interested reader to the detailed introduction in the work of Andrews \& Wigderson \cite{AW24} for further pointers to this literature. 

Our focus in this work is to study the parallelizability of GCD computation of two univariate polynomials. More formally, we are interested in the following problem -- can the GCD of two polynomials be computed by a polynomial size constant depth algebraic circuit that takes the coefficients of the polynomials as input? A recent work of Andrews \& Wigderson \cite{AW24} showed that over fields of characteristic zero  or sufficiently large characteristic, this is indeed the case, i.e. the GCD of two univariate polynomials can be computed (piece-wise) by constant depth algebraic circuits of polynomial size. Here, the ``piece-wise'' part is to account for the fact that the GCD of two polynomials is not a rational function of their coefficients. We define this model of piece-wise algebraic circuits formally in \cref{defn:piece-wise-rational-fn}, but for this introduction, we encourage the reader to think of this as a standard algebraic circuit. In addition to the GCD, Andrews \& Wigderson showed a similar constant depth upper bound for many closely related algebraic problems like the LCM, the resultant, the discriminant, the inverse of the Sylvester matrix and computing square-free decomposition of a polynomial. 

The key technical fact underlying the results in \cite{AW24} was that many of these algebraic problems have a relatively natural constant depth algebraic circuit if the input polynomials were given via their roots, and not their coefficients. Moreover, these circuits really compute functions that are symmetric in the roots of the inputs polynomials. Thus, if one could somehow compute these symmetric functions of roots via a constant depth algebraic circuit that takes as input the coefficients of the polynomials (and not their roots), we would be done. Over fields of zero or sufficiently large characteristic, Andrews \& Wigderson showed that this is indeed the case. Their primary technical ingredient for this is the fact that over such fields, there are constant depth algebraic circuits of polynomial size that can transform the first $n$ power-sum symmetric polynomials to the $n$ elementary symmetric polynomials and vice-versa. Unfortunately, this last fact is not true when the underlying field has small (non-zero) characteristic. In fact, over such fields, power-sum symmetric polynomials of degree $1, 2, \ldots, n$ are not algebraically independent, while the elementary symmetric polynomials are, and hence we cannot expect to compute the $n$ elementary symmetric polynomials from the first $n$ power-sum symmetric polynomials. While this specific issue is fixable to an extent by considering power-sum symmetric polynomials of carefully chosen degrees (e.g. over fields of characteristic $2$, we take the first $n$ power-sums of odd degree), it is unclear if the transformation from this set of power-sums to elementary symmetric polynomials can be done efficiently in constant depth. This technical issue obstructs the natural attempt at extending the results in \cite{AW24} to finite fields of small characteristic and necessitates the need for some new technical ideas.

In this work, we show that the GCD of two univariates can indeed be computed by polynomial size constant depth algebraic circuits (again, piece-wise) over all fields of polynomially large size, but arbitrary characteristic. We now discuss our results in more detail, before moving on to an overview of the proofs. 

\subsection{Our results}

\begin{theorem}\label{thm:intro-gcd}
There is an absolute constant $c \in \N$ such that the following is true. 

For all $n \in \N$ and any field $\F$ with at least $n^c$ elements, the $\gcd$ of two univariates over $\F$ of degree at most $n$ can be computed piece-wise as a ratio of algebraic circuits of size $\poly(n)$ and depth $O(1)$. 
\end{theorem}

A more detailed version of the above theorem statement is provided in \cref{thm:gcd-computation}. The above (piece-wise) circuit yields an efficient parallel algorithm for computing the GCD of two polynomials. More precisely, there is an algebraic PRAM algorithm with polynomially many processors and constant number of rounds that takes the coefficient vectors of two univariate polynomials as input and computes their GCD. Each processor is allowed to perform basic field operations and branching instructions (based on if a variable is 0 or not).

Our proof techniques also extend to some of the other results  in \cite{AW24} and show that these continue to be true over all large enough fields. Perhaps the most interesting among these is that there is a constant depth algebraic circuit of $\poly(n)$ size that takes as input the coefficients of two univariates of degree $n$ and computes their resultant. From \autoref{thm:intro-gcd}, we also get a piece-wise small constant-depth circuit for the LCM of two polynomials using the fact that $\mathsf{LCM}(f,g) = \frac{f\cdot g}{\gcd(f,g)}$. However, our techniques do not extend to many of the other results in \cite{AW24}, most notably that of computing square-free decomposition and results on arbitrary manipulation of the multiplicities of roots (Section 9 in the full version of \cite{AW24}). 

On our way to the proof of \autoref{thm:intro-gcd}, we prove the following theorem that might be independently interesting. 
\begin{theorem}[Complexity of symmetric polynomials]\label{thm:intro-complexity-of-symmetric-polys}
    Let $n, d, s \in \N$ be parameters, and let $\F$ be a sufficiently large field, i.e. $|\F| \geq (nds)^c$ for an absolute constant $c > 0$. Let $P(\vecx) \in \F[\vecx]$ be a symmetric polynomial on $n$ variables of degree $d$ and let $Q(\vecz) \in \F[\vecz]$ be the unique degree $d$ polynomial such that $P(\vecx) = Q(\esym_1(\vecx), \dots, \esym_n(\vecx))$. Then, the following are true: 
    \begin{itemize}\itemsep 0pt
        \item If $P$ is computable by a circuit size $s$ and depth $\Delta$, then $Q$ can be computed by a circuit of size $\poly(s,d,n)$ and depth $\Delta + O(1)$. 

        \item If $P$ is computable by an algebraic formula of size $s$, then $Q$ can be computed by a formula of size $\poly(s,d,n)$. 
    \end{itemize}
\end{theorem}
In \cite{BJ19}, Bl\"{a}ser \& Jindal showed that if an $n$-variate symmetric polynomial $P$ can be computed by an algebraic circuit of size $s$ and degree $d$, then it can be expressed as the composition of a circuit $Q$ of size and degree $\poly(s,d, n)$ with elementary symmetric polynomials. However, if $P$ is assumed to be computable by a constant-depth circuit or a formula, it is unclear from the proof in \cite{BJ19} whether $Q$ also has a small constant-depth circuit or a formula respectively. Thus, \autoref{thm:intro-complexity-of-symmetric-polys} is a version of the result of \cite{BJ19} for constant depth algebraic circuits and formulas. Showing such a statement for constant-depth circuits was mentioned as an open problem in \cite{AW24} and \autoref{thm:intro-complexity-of-symmetric-polys} answers this question affirmatively. 

Our proofs are based on techniques that were recently used to show the closure of constant-depth circuits and formulas under factorization in \cite{BKRRSS25} over fields of zero or large characteristic. These closure results, in turn, are a consequence of a classical theorem of Furstenberg \cite{Furstenberg67} that gives a computationally simple and explicit expression for the power series roots of a bivariate polynomial (see \autoref{thm:furstenberg} and \autoref{thm:furstenberg-small-characteristic}). 

Our proof techniques have another consequence of closure under factorization over fields of small characteristic. In \cite{BKRRSS25}, while the complete closure result was over fields of zero or large characteristic, a weaker statement was shown to be over fields of small characteristic. More specifically, it was shown (Theorem 1.3 in \cite{BKRRSS25}) that if an $n$ variate degree $d$ polynomial $f$ has a constant-depth circuit of size $s$ over the field $\F_{p^k}$ (for a prime $p$) and $g$ is a factor of $f$, then for some $i \in \N$, $g^{p^i}$ has a constant-depth circuit of size $\poly(s, n, d)$ over the algebraic closure (or a very high degree extension) of $\F_{p^k}$. Using the techniques here, we show that the circuit for $g^{p^i}$ is in fact over the field $\F_{p^k}$ itself, provided that this field is polynomially large. We refer to \autoref{thm:closure-factors-small-char} for a formal statement of this result.

\subsection{Proof overview}
\paragraph*{Proof of \autoref{thm:intro-complexity-of-symmetric-polys}:}

We start with a discussion of the proof of \autoref{thm:intro-complexity-of-symmetric-polys}. To this end, we begin with an outline of the proof in \cite{BJ19}, who originally showed the version of \autoref{thm:intro-complexity-of-symmetric-polys} for general algebraic circuits.  Here the authors consider the polynomial \[
F(\vecz, y) = y^n - z_1 y^{n-1} + \dots + (-1)^nz_n  - 1  ,
\]
and think of the polynomial $F$ as a univariate in $y$ with coefficients from the ring $\F[\vecz]$. Since $F(\mathbf{0}, y) = y^n -1 $ has $n$ distinct roots, namely the $n^\text{th}$ roots of unity, we get by an application of Newton Iteration that each of these roots of $y^n - 1$ can be \emph{lifted} to a unique power series (in $\vecz$) root of $F(\vecz, y)$. Let these power series roots be $A_0(\vecz), A_1(\vecz), \ldots, A_{n-1}(\vecz)$. If $P(\vecx)$ is a symmetric polynomial, then there is a unique $n$ variate polynomial $Q$ of degree at most $\deg(P)$ such that $P(\vecx) = Q(\esym_1(\vecx), \ldots, \esym_n(\vecx))$, where $\esym_i(\vecx)$ is the elementary symmetric polynomial of degree $i$ in $\vecx$. Let us consider the substitution where the  variable $x_i$ is replaced by the power series $A_i(\vecz)$. Note that $\esym_j(A_0, \ldots, A_{n-1})$ must equal $z_j$ for every $j \in \{1, 2, \ldots, n-1\}$ and equals $z_n - 1$ for $j = n$.  
So, we get 
\[
P(A_0(\vecz), \dots, A_{n-1}(\vecz)) = Q(z_1, \ldots, z_n - 1) \, .
\]
We now note that since $Q$ is a polynomial, it suffices to compute the power series $A_i(\vecz)$ modulo $\langle \vecz \rangle^{\deg(Q) + 1}$. Using the fact that if $F$ has a small algebraic circuit, all its truncated power series roots have small circuits, and with a few small modifications, we get a small unbounded depth circuit for $Q$. 

Given this outline, to prove \autoref{thm:intro-complexity-of-symmetric-polys}, it would suffice to show that the truncated power series roots of $F(\vecz, y)$ have small constant-depth circuits, and we do exactly this! Note that $F$ itself has a small constant-depth circuit.  We now invoke the closure results for factors of constant-depth circuits shown in \cite{BKRRSS25} (Theorem 4.1) to show that the truncations of the power series $A_i$ can be computed by small constant-depth circuits. This fact together with the above outline gives us \autoref{thm:intro-complexity-of-symmetric-polys}. 

We note that even though the main closure results in \cite{BKRRSS25} are for fields of zero or large characteristic, the instance we have at hand here is very special --- we are trying to compute the power series roots of a polynomial that has no repeated roots. For such instances, essentially the arguments from fields of characteristic zero carry over as is to fields of small positive characteristic. This completes an outline of the proof of \autoref{thm:intro-complexity-of-symmetric-polys}. 

\paragraph*{Computing the resultant in constant depth:}
To prove \autoref{thm:intro-gcd}, we first observe that \autoref{thm:intro-complexity-of-symmetric-polys} can be generalized for polynomials that are bi-symmetric. More precisely, we observe that if  a $2n$ variate polynomial $P(\vecx, \vecy)$  is symmetric in $\vecx$ and in $\vecy$ (although it need not be symmetric in $\vecx \union \vecy$ together), then $P$ equals $Q(\esym_1(\vecx), \dots, \esym_n(\vecx), \esym_1(\vecy), \ldots, \esym_n(\vecy))$ for a unique polynomial $Q$. We then show that if $P$ has a small constant-depth circuit, then so does $Q$. 

An almost immediate (and independently) interesting consequence of the above discussion is that the resultant of two monic univariates can be computed by a constant-depth circuit that takes the coefficients of the polynomials as inputs. If $f(x) = f_0 + f_1x + \dots + x^n$ and $g(x) = g_0 + g_1x + \dots +  x^n$ are two univariates with roots denoted by $\{\alpha_1, \ldots, \alpha_n\}$ and $\{\beta_1, \ldots, \beta_n \}$ respectively, then their resultant $R$ equals $\prod_{i, j \in [n]} (\alpha_i - \beta_j)$. We note that $R$ is a polynomial that is symmetric in the roots of $f$ and the roots of $g$, and has a small circuit of constant depth that takes these roots as inputs. Thus, from the generalization of \autoref{thm:intro-complexity-of-symmetric-polys} to bisymmetric polynomials, we get that there is a polynomial $Q$ that can be computed by a small constant-depth circuit, such that $Q$ composed with the elementary symmetric polynomials of $\{\alpha_1, \ldots, \alpha_n\}$ -- the coefficients of $f$ -- and the elementary symmetric polynomials of $\{\beta_1, \ldots, \beta_n\}$ -- the coefficients of $g$ -- equals $R$. Thus, we have a small constant-depth circuit for the resultant of two univariates that takes the coefficients as input. 

\paragraph{Computing the GCD:} It would be helpful to revisit the proof strategy of Andrews and Wigderson~\cite{AW24}. If the polynomials $f$ and $g$ are square-free, then the $\gcd(f,g)$ can be obtained by `filtering' from $f$ only those roots of $f$ that are also roots of $g$. One of the key aspects of \cite{AW24} was to execute this `filtering' operation using the efficient transformation (via constant-depth circuits) between elementary and power-sum symmetric polynomials of the roots of $f$ and $g$. In the more general setting when $f$ and $g$ are not square-free, using some additional ideas, \cite{AW24} reduce to the square-free setting by first computing the \emph{square-free decomposition}\footnote{i.e., computing $f_1,\ldots, f_d$ that are square-free and co-prime such that $f = \prod f_i^i$} and working with each part in the decomposition. 

In the setting of characteristic $p > 0$, we have to address two issues --- (a) we do not have efficient transformations between power-sum and elementary symmetric polynomials, and (b) computing the square-free part of $f$ is a more general operation than computing the $p$-th root of polynomials (something that we still do not know how to do efficiently and this would be very interesting). We give an alternate implementation of the filtering operation using \cref{thm:intro-complexity-of-symmetric-polys} and also have a simple observation that allows us to completely bypass the calculation of the square-free part. We briefly elaborate on the key steps below. 

\begin{itemize}
    \item \textbf{An auxiliary polynomial containing $\gcd(f,g)$:} Our first observation is that, if we consider a \emph{generic} linear combination of $f$ and $g$ given by $F(y,z) := f(y)+z \cdot g(y)$, then each root of $\gcd(f,g)$ (from $\overline{\F}$) divides $F$ with the \emph{right} multiplicity. Furthermore, there is no other $\alpha \in \overline{\F}$ satisfying $F(\alpha,z) = 0$. Thus, if we could somehow \emph{filter} the $y$-roots of $F(y,z)$ with their multiplicity, we will obtain $\gcd(f,g)$. This thereby avoids computing square-free components of $f$ and $g$ entirely. 
    
    \item \textbf{Filtering roots of one polynomial using another, à la \cite{AW24}:} Suppose $\set{\sigma_1,\ldots, \sigma_m}$ is the multi-set of roots of $f$ over $\overline{\F}$ and $g(y) \in \F[y]$, \cite{AW24} define a `filter' operation to compute 
    \[
    \Filter(f \mid g \neq 0) = \prod_{i\in [m] \;:\; g(\sigma_i) \neq 0} (y - \sigma_i).
    \]
    \cite{AW24} show that the above can essentially be computed piece-wise via elementary symmetric polynomials of $\set{g(\sigma_1), \ldots, g(\sigma_m)}$ with the advice parameter being the degree of the filter. Rather than proceeding via power-sum symmetric polynomials (which requires large characteristic fields), we once again use \cref{thm:intro-complexity-of-symmetric-polys} to express this as a constant-depth circuit over the coefficients of $f$ and $g$. \cref{thm:intro-gcd} now follows by noticing that 
    \[
    \gcd(f(y),g(y)) = \Filter(F(y,z) \mid g = 0) = \frac{F(y,z)}{\Filter(F(y,z) \mid g \neq 0)}.
    \]
\end{itemize}

\subsection*{Organization of the paper}
The rest of the paper is organized as follows. We start with some general notations and preliminaries in \autoref{sec:prelims}, followed by the proof of \autoref{thm:intro-complexity-of-symmetric-polys} in \autoref{sec:symmetric}. In \autoref{sec:GCD}, we show that the resultant and GCD can be computed by constant-depth circuits  and in \autoref{sec:closure-small-char}, we use these techniques to strengthen the closure results for constant-depth circuits under factorization. 

\section{Preliminaries}
\label{sec:prelims}

\paragraph{Notation}
\begin{itemize}\itemsep 0pt
    \item We use $\F, \K$ etc. to refer to fields and $\overline{\F}$ to refer to the algebraic closure of $\F$. Also, $\F[x]$ refers to the polynomial ring, $\F\indsquare{x}$ to the ring of formal power series, and $\F\indparen{x}$ refer to ring of Laurent series with respect to the variable $x$ with coefficients from the field $\F$. 
    \item We use boldface letters such as $\vecx$ to refer to an ordered tuple of variables
    \item Given a univariate polynomial $f(x) \in \F[x]$ (or more generally in the ring of formal Laurent series $\F\indparen{x}$), we denote by $\coeff{x^n}{f}$ the coefficient of the monomial $x^n$ in $f$. In the case of a multivariate polynomial, such as $f(x, y)$, we interpret $f$ as an element of $\F[y][x]$—that is, as a univariate polynomial in $x$ with coefficients in $\F[y]$—and $\coeff{x^n}{f}$ then refers to the $x^n$ coefficient viewed as a polynomial in $y$.
    \item The notation $\homog_{d}(F)$ refers to the degree $d$ homogeneous part of $F$, and $\homog_{\leq d}(F)$ refers to the sum of all homogeneous parts of $F$ up to degree $d$. For any subset $\vecx_S \subseteq \vecx$ and any $i \in [d]$, the degree $i$ homogeneous part of $F$ with respect to $\vecx_S$, denoted by $\homog_{\vecx_S, i}(F)$. 
    \item The $i^{th}$ elementary symmetric polynomial is sum of all multilinear monomials in $\vecx$ of degree $i$. Formally, the $i^{th}$ elementary polynomial is $\esym_i(\vecx) \coloneq \sum_{S \subseteq [n] :\abs{S}=i} \prod_{i \in S} x_i$.
\end{itemize}

\subsection{Interpolation}

The following uses of polynomial interpolation in the context of algebraic circuits are due to Michael Ben-Or.

\begin{lemmawp}[Interpolation]
    \label{lem:interpolation}
    Let $\F$ be a subfield of a commutative ring $R$, and suppose that $\F$ contains at least $d+1$ distinct elements. Choose distinct points $\alpha_0, \ldots, \alpha_d \in \F$. Then, for any polynomial $f(t) = f_0 + f_1 t + \cdots + f_d t^d \in R[t]$ of degree at most $d$, and for each index $i \in \{0, \ldots, d\}$, there exist elements $\beta_{i0}, \ldots, \beta_{id} \in \F$ (depending only on $i$ and the $\alpha_j$'s) such that the coefficient $f_i$ can be recovered from the values of $f$ at the interpolation points:
    \[
        f_i = \beta_{i0} f(\alpha_0) + \beta_{i1} f(\alpha_1) + \cdots + \beta_{id} f(\alpha_d).\qedhere
    \]
\end{lemmawp}

\begin{corollarywp}[Applications of interpolation]
    \label{cor:interpolation-consequences}
    Let $\alpha_0, \ldots, \alpha_d$ be distinct elements of the field $\F$. The following statements are immediate consequences of interpolation:

    \begin{enumerate}\itemsep0pt
        \item \textbf{[Extracting homogeneous parts]}  
        Let $C(\vecx)$ be a polynomial of total degree $d$, and let $\vecx_S \subseteq \vecx$ be a subset of the variables. For any $i \in \{0, \ldots, d\}$, the degree $i$ homogeneous component of $C$ with respect to the variables in $\vecx_S$, denoted $\homog_{\vecx_S, i}(C)$, can be written as
        \[
            \homog_{\vecx_S, i}(C) = \sum_{j=0}^d \beta_{i,j} \cdot C(\alpha_j \cdot \vecx_S, \vecx_{\overline{S}})
        \]
        for some scalars $\beta_{i,j} \in \F$ that are independent of the polynomial $C$.

        \item \textbf{[Derivatives via evaluation]}  
        Suppose $C(\vecx, y)$ is a polynomial with degree $d$ in the variable $y$. Then the $i^{th}$ partial derivative of $C$ with respect to $y$ can be expressed as an $\F[y]$-linear combination of the evaluations $C(\vecx, \alpha_j)$ for $j = 0, \ldots, d$. Specifically, there exist polynomials $\mu_0(y), \ldots, \mu_d(y) \in \F[y]$, each of degree at most $d$ and independent of $C$, such that
        \[
            \partial_{y^i} C(\vecx, y) = \sum_{j=0}^d \mu_j(y) \cdot C(\vecx, \alpha_j).
        \]
    \end{enumerate}

    Moreover, if $C$ is computed by a circuit of size $s$ and depth $\Delta$, then both operations—extracting homogeneous components and computing partial derivatives—can be implemented by circuits of size $\poly(s, d)$ and depth at most $\Delta + O(1)$.
\end{corollarywp}

\subsection{Resultants}

\begin{definition}[Sylvester Matrix and Resultant] \label{def:Sylvester-Resultant}
    Let $\F$ be a field. Let $P(z)$ and $Q(z)$ be polynomials of degree $a\geq 1$ and $b \geq 1$ in $\F[z]$. Define a linear map $\Gamma_{P,Q}:\F^a \times \F^b \to \F^{a+b}$ that takes polynomials $A(z)$ and $B(z)$ in $\F[z]$ of degree $a-1$ and $b-1$ respectively, and maps them to $AP + BQ$, a polynomial of degree $a+b-1$.

    The \emph{Sylvester matrix} of $P$ and $Q$, denoted by $\operatorname{Syl}_z(P,Q)$, is defined to be the $(a+b)\times(a+b)$ matrix for the linear map $\Gamma_{P,Q}$. 

    The \emph{Resultant} of $P$ and $Q$, denoted by $\Res{z}{P}{Q}$, is the determinant of $\operatorname{Syl}_z(P,Q)$.
\end{definition}

\begin{fact}[see, e.g., Chapter 3 of \cite{CLO05}]
    \label{fact:resultant-product-of-root-differences}
    Suppose $f(z), g(z) \in \F[z]$ are monic polynomials, and the multiset of their roots over $\overline{\F}$ are $\alpha_1,\ldots, \alpha_a$ and $\beta_1,\ldots, \beta_b$. Then,
    \[
    \Res{z}{f}{g} = \prod_{i=1}^a \prod_{j=1}^b (\alpha_i - \beta_j).
    \]
\end{fact}

\subsection{Factorization of monic polynomials into power series}

\begin{lemma}[Factorization into power series]
    \label{lem:factorisation-into-power-series}
    Let $f(t, y) \in \K[t,y]$ be a polynomial that is monic in $y$ such that $f(0, y)$ is square-free. For each $\alpha \in \overline{\K}$ (the algebraic closure) such that $f(0, \alpha) = 0$, there is a unique power series $\phi_\alpha(t) \in \overline{\K}\indsquare{t}$ satisfying $\phi_\alpha(0) = \alpha$ such that $f(t, \phi_\alpha(t)) = 0$. 

    \noindent
    In fact, the polynomial $f(t, y)$ factorizes in $\overline{\K}\indsquare{t}[y]$ as 
    \[
    f(t, y) = \prod_{\alpha \in Z}(y - \phi_\alpha(t))
    \]
    where $Z$ is the set of roots of $f(0, y)$ in $\overline{\K}$. 
\end{lemma}

\noindent
The above lemma is essentially folklore and \cite[Section 3]{DSS22-closure} gives a formal proof of the above. We also note that the lemma can be applied in the multivariate setting by taking a multivariate polynomial in variables $x_1, x_2, \ldots, x_n, y$, replacing each $x_i$ by $x_i \cdot t$ (for a new indeterminate $t$) and viewing the resulting polynomial as a bivariate in $t$ and $y$ with coefficients in the field $\F(x_1, x_2, \ldots, x_n)$.

\subsection{Furstenberg's theorem for roots of multiplicity one}

We define the following \emph{diagonal} operator for bivariate power series $F(x,y) = \sum_{i,j} F_{i,j} x^i y^j\in \F\indsquare{x,y}$ as follows:
\[
\diag(F)(t) = \sum_{i\geq 0} F_{i,i} \cdot t^i.
\]

\begin{theorem}[{\cite[Proposition 2]{Furstenberg67}}] \label{thm:furstenberg}
    Let $P(t,y) \in \F\indsquare{t,y}$ be a power series and $\varphi(t){\in \F\indsquare{t}}$ be a power series satisfying $P(t, y) = (y - \phi(t)) \cdot Q(t,y)$. If $\phi(0) = 0$ and $Q(0,0) \neq 0$, then
    \begin{equation}
        \varphi = \diag\inparen{\frac{y^2 \cdot \partial_yP(ty,y)}{P(ty,y)}}
    \end{equation}
\end{theorem}

The diagonal expression above can be simplified to a slightly more convenient expression for implicitly defined power series roots. 

\begin{corollary}
    \label{cor:flajolet-soria-formula-for-roots}
    Let $P(t,y) \in \F\indsquare{t,y}$ and $\phi(t) \in \F\indsquare{t}$ such that $\phi(0) = 0$, $P(t, \phi(t)) = 0$ and $\partial_y P(0,0) = \alpha \neq 0$. Then, 
    \[
    \phi(t) = \sum_{m \geq 1} \frac{1}{\alpha^{m+1}} \cdot \coeff{y^{m-1}}{\partial_y P(t,y) \cdot (\alpha y - P(t,y))^m}.
    \]
\end{corollary}


Suppose $P(\vecx, y)$ is computable by a small constant-depth circuit. Using \autoref{cor:flajolet-soria-formula-for-roots}, \cite{BKRRSS25} showed that truncations of power series roots (of multiplicity 1) of $P(\vecx, y)$ also have small constant-depth circuits. 

\begin{theorem}[Power series roots without multiplicity \cite{BKRRSS25}]
    \label{thm:closure-powerseries-roots}
    Let $P(\vecx,y) \in \F[\vecx,y]$ be a polynomial computed by a circuit $C$, and let $\varphi(\vecx) {\in \F\indsquare{\vecx}}$ be a power series satisfying $\varphi(\veczero) = 0$, $P(\vecx,\varphi(\vecx)) = 0$ and $\partial_y P(\veczero,0) \neq 0$. Then, for any $d \in \N$, there is a circuit $C'$ computing $\homog_{\leq d}\insquare{\varphi}$ such that
    \[\size(C') \leq \poly(d,\size(C))\]
    \[\depth(C') \leq \depth(C) + O(1)\]
\end{theorem}

\section{Complexity of Symmetric Polynomials}\label{sec:symmetric}

A polynomial $P(x_1, \dots, x_n)$ is symmetric if for every permutation $\sigma$ on $n$ elements, $P(x_1, \dots, x_n) = P(x_{\sigma(1)}, \dots, x_{\sigma(n)})$. The elementary symmetric polynomials $\{\esym_i(\vecx)\}_{i\in [n]}$ are ubiquitous examples of symmetric polynomials. Surprisingly, the \emph{fundamental theorem of symmetric polynomials} states that every symmetric polynomial $P$ can be uniquely expressed as polynomial expressions in $\{\esym_i(\vecx)\}_{i\in [n]}$. 

\begin{theorem}[Fundamental theorem of symmetric polynomials] \label{thm:fundamental-theorem-of-symm-poly}
    Let $P(\vecx)$ be a symmetric polynomial on $n$ variables of degree $d$. Then there is a unique polynomial $Q(y_1, \dots, y_n)$ such that $P(\vecx) = Q(\esym_1(\vecx), \esym_2(\vecx), \dots, \esym_n(\vecx))$. Moreover, the degree of $Q$ is at most $d$.
\end{theorem}

Given this fact, it is natural to ask if complexity of the polynomial $Q$ is related to the complexity of $P$? Bl\"{a}ser and Jindal~\cite{BJ19} showed that if $P$ is computable by a small algebraic circuit, then $Q$ is also computable by a small algebraic circuit.

\begin{theorem}[Theorem 4 in \cite{BJ19}] \label{thm:general-ckt-complexity-of-symm-poly}
    Let $P(\vecx) \in \Complex[\vecx]$ be a symmetric polynomial on $n$ variables of degree $d$ computed by a circuit $C$ of size $s$. Let $Q(\vecx) \in \Complex[\vecx]$ be the unique polynomial such that ${P(\vecx) = Q(\esym_1(\vecx), \dots, \esym_n(\vecx))}$.
    Then, $Q(\vecx)$ is computable by a circuit of size $\poly(s,d,n)$ over $\Complex$.
\end{theorem}

One of the main steps in their proof is to compute the truncations of power series roots using Newton Iteration. Based on the ideas from \cite{BKRRSS25}, we would expect that replacing Newton Iteration with Furstenberg's theorem (\cref{thm:furstenberg}) should give us an analogous result over constant-depth circuits. Indeed, this turns out to be the case. We now prove the constant-depth version of their result.

\begin{theorem}[Complexity of symmetric polynomials]
    \label{thm:complexity-of-symmetric-polys}
    Let $n, d, s \in \N$ be parameters, and let $\F$ be a polynomially large field, i.e. $|\F| \geq (nds)^c$ for an absolute constant $c > 0$. Let $P(\vecx) \in \F[\vecx]$ be a symmetric polynomial on $n$ variables of degree $d$ computed by a circuit $C$ of size $s$ and depth $\Delta$. 
    If $Q(\vecz) \in \F[\vecz]$ is the unique degree $d$ polynomial such that $P(\vecx) = Q(\esym_1(\vecx), \dots, \esym_n(\vecx))$, then $Q(\vecz)$ is computable by a circuit of size $\poly(s,d,n)$ and depth $\Delta + O(1)$ over $\F$.
\end{theorem}
\begin{proof}
    Let $\alpha_1, \dots, \alpha_n$ be $n$ distinct elements from $\F$ and $\beta_i$ denote $\esym_i(\alpha_1, \dots, \alpha_n)$ for each $i \in [n]$. Define $F(\vecz,y)$ as: 
    \[
        F(\vecz,y) := y^n - (z_1 + \beta_1) y^{n-1} + \cdots + (-1)^{n-1} (z_{n-1} + \beta_{n-1}) y + (-1)^n (z_n + \beta_n).
    \]
    Note that $F(\veczero,y) = y^n + \sum_{i=1}^n {(-1)^iy^{n-i}\beta_i} = \prod_{i\in[n]}(y-\alpha_i)$. Since $\alpha_1,\ldots, \alpha_n$ are distinct, we have that $F(\veczero, y)$ is square-free and hence (by \cref{lem:factorisation-into-power-series}) 
    $F(\vecz, y)$ factorizes as
    \[
    F(\vecz, y) = \prod_{i=1}^n (y - A_i(\vecz))
    \]
    where $A_i(\vecz) \in \F\indsquare{\vecz}$ with $A_i(\veczero) = \alpha_i$ for each $i \in [n]$. Furthermore, $\esym_k(A_1(\vecz), \ldots, A_n(\vecz)) = z_k + \beta_k$. 

    We now apply \cref{thm:closure-powerseries-roots} on $F(\vecz,y)$ to get size $\poly(s,d,n)$ and depth $\Delta+O(1)$ circuits $\tilde{A}_1, \dots, \tilde{A}_n$ such that $\tilde{A}_i(\vecz) = A_i(\vecz) \bmod \inangle{\vecz}^{d+1}$ for each $i\in [n]$.  This implies that for each $i \in [n]$, 
    \begin{align*}
    \esym_i(\tilde{A}_1, \dots, \tilde{A}_n) & = z_i + \beta_i \bmod \inangle{\vecz}^{d+1}.\\
    \therefore \quad P(\tilde{A}_1, \dots, \tilde{A}_n) & = P(A_1, \dots, A_n) \bmod{\inangle{\vecz}^{d+1}}
    \end{align*}
    From \cref{thm:fundamental-theorem-of-symm-poly}, degree of $Q$ is at most $d$. Therefore,
    \begin{align*}
    P(\tilde{A}_1, \dots, \tilde{A}_n) & = P(A_1, \dots, A_n) \bmod{\inangle{\vecz}^{d+1}} \\
        & = Q(\esym_1(A_1,\ldots, A_n), \ldots, \esym_n(A_1,\ldots, A_n))\\
        & = Q(z_1 + \beta_1, \ldots, z_n + \beta_n)\\
    \implies Q(z_1 + \beta_1, \ldots, z_n+\beta_n) & = \homog_{\leq d} P(\tilde{A}_1, \dots, \tilde{A}_n)
    \end{align*}
    Since $P$ and each $\tilde{A}_i$ is computable by a size $\poly(s,d,n)$ and depth $\Delta + O(1)$ circuit, we obtain a $\poly(s,d,n)$ and depth $\Delta + O(1)$ circuit for $Q(\vecz + \vecbeta)$, and thus for $Q(\vecz)$ as well (by shifting $\vecz \leftarrow \vecz - \vecbeta$). 
\end{proof}

\subsection{Multi-symmetric polynomials}

\begin{definition}[Multi-symmetric polynomials] Let $\vecx = \vecx_1 \sqcup \cdots \sqcup \vecx_k$ be a partition of the set of variables $\vecx$. A polynomial $f(\vecx) = f(\vecx_1,\ldots, \vecx_k)$ is said to be \emph{multi-symmetric} with respect to the above partition if for all $i \in [k]$ we have that $f(\vecx)$ is symmetric with respect to the part $\vecx_i$. In other words, for every choice of permutations $\pi_1, \dots, \pi_k$ on $\vecx_1,\dots, \vecx_k$, we have $f(\pi_1(\vecx_1), \dots, \pi_k(\vecx_k)) = f(\vecx_1,\ldots, \vecx_k)$. 
\end{definition}

The following is a natural extension of \cref{thm:fundamental-theorem-of-symm-poly}. For purely the ease of exposition, we assume the size of each part in the partition is the same; the proof would readily extend to the general case with mere notational changes. 

\begin{theorem}[Fundamental theorem for multi-symmetric polynomials] \label{thm:fundamental-theorem-of-multi-symm-poly}
    Let $\vecx = \vecx_1 \sqcup \cdots \sqcup \vecx_k$ be a partition of the variables with each $\abs{\vecx_i} = n$, and let $P(\vecx)$ be a multi-symmetric polynomial with respect to the above partition of degree $d$. Then there is a unique polynomial $Q(\vecy_1,\dots, \vecy_k)$ (where $\vecy_i = (y_{i1}, \ldots, y_{in})$) such that $P(\vecx) = Q(\overline{\esym}(\vecx_1), \dots, \overline{\esym}(\vecx_k))$ where 
    \[
    \overline{\esym}(\vecx_i) \coloneq \inparen{\esym_1(\vecx_i), \ldots, \esym_n(\vecx_i)}. 
    \]
    Moreover, the degree of $Q$ is at most $d$.
\end{theorem}
\begin{proof}
The proof is just a simple induction on $k$. For $k = 1$, the claim follows from \cref{thm:fundamental-theorem-of-symm-poly}. 

Assuming that the theorem is true for $k-1$, and write $P(\vecx)$ as 
\[
P(\vecx) = \sum_{\vece \geq 0} \coeff{\vecx_k^{\vece}}{P} \cdot \vecx_k^\vece.
\]
We will say $\vece_i \sim \vece_j$ if there is some permutation $\pi$ on $\abs{\vecx_k}$ elements such that $\pi(\vece_i) = \vece_j$; let $\mathcal{E}$ be a distinct set of representatives of the equivalence classes defined by $\sim$. Note that if $\vece_i \sim \vece_j$, then $\coeff{\vecx_k^{\vece_i}}{P} = \coeff{\vecx_k^{\vece_j}}{P}$. Therefore, if we define $H_{\vece^*}(\vecx_k)$ for an $\vece^* \in \mathcal{E}$ as
\[
H_{\vece^*}(\vecx_k) \coloneq \sum_{\vece'\;:\; \vece' \sim \vece^*} \vecx_k^{\vece'}
\]
then, 
\begin{align*}
P(\vecx) &= \sum_{\vece^* \in \mathcal{E}} \coeff{\vecx_k^{\vece^*}}{P} \cdot H_{\vece^*}(\vecx_k) \\
         &= \sum_{\vece^* \in \mathcal{E}} \coeff{\vecx_k^{\vece^*}}{P} \cdot \tilde{Q}_{\vece^*}(\overline{\esym}(\vecx_k)) & \text{(by \cref{thm:fundamental-theorem-of-symm-poly})}
\end{align*}
for some $\tilde{Q}_{\vece^*}$ of degree at most $\abs{\vece^*}$. 
Since $P$ is multi-symmetric, we have that $\coeff{\vecx_k^{\vece}}{P}$ is multi-symmetric with respect to $\vecx_1 \sqcup \cdots \sqcup \vecx_{k-1}$.  By induction, there exists polynomials $Q_{\vece}$ of degree at most $d - \abs{\vece}$ such that 
\begin{align*}
\coeff{\vecx_k^{\vece}}{P} & = Q_\vece(\overline{\esym}(\vecx_1), \dots, \overline{\esym}(\vecx_{k-1})). \\
\implies P(\vecx) & = \sum_{\vece^* \in \mathcal{E}} Q_{\vece^*}(\overline{\esym}(\vecx_1), \dots, \overline{\esym}(\vecx_{k-1})) \cdot \tilde{Q}_{\vece^*}(\overline{\esym}(\vecx_k))\\
 & =: Q(\overline{\esym}(\vecx_1), \dots, \overline{\esym}(\vecx_{k}))
\end{align*}
for some $Q(y_{11}, \ldots, y_{kn})$ of degree at most $d$. 
\end{proof}

\begin{theorem}[Complexity of multi-symmetric polynomials]
    \label{thm:complexity-of-multi-symmetric-polys}
    Let $\F$ be a polynomially large field with respect to parameters $n, d, s \in \N$, i.e. $|\F| \geq (nds)^c$ for an absolute constant $c > 0$. Let $\vecx = \vecx_1 \sqcup \dots \sqcup \vecx_k$ be a partition of the variables, with each $\abs{\vecx_i} = n$, and let $P(\vecx) \in \F[\vecx]$ be a multi-symmetric polynomial with respect to the partition of degree $d$ computed by a circuit $C$ of size $s$ and depth $\Delta$. 
    If $Q(z_{11}, \ldots, z_{kn}) \in \F[\vecz]$ is the unique degree $d$ polynomial such that 
    \[
        P(\vecx) = Q(\overline{\esym}(\vecx_1), \dots, \overline{\esym}(\vecx_k))
    \]
    where $\overline{\esym}(\vecx_i) \coloneq \inparen{\esym_1(\vecx_i), \ldots, \esym_n(\vecx_i)}$, then $Q(\vecz)$ is computable by a circuit of size $\poly(s,d,n)$ and depth $\Delta + O(1)$ over $\F$.
\end{theorem}
\begin{proof}
The proof is almost a direct extension of \cref{thm:complexity-of-symmetric-polys}. As in the proof of \cref{thm:complexity-of-symmetric-polys}, for each $i \in [k]$ define the polynomials for $\vecz_i = (z_{i,1}, \ldots, z_{i,n})$:
\[
    F_i(\vecz_i,y) := y^n - (z_{i,1} + \beta_1) y^{n-1} + \cdots + (-1)^{n-1} (z_{i,n-1} + \beta_{n-1}) y + (-1)^n (z_{i,n} + \beta_n)
\]
where $\beta_i = \esym_i(\alpha_1,\ldots, \alpha_n)$ for some distinct choice of $\alpha_1,\ldots, \alpha_n \in \F$. By \cref{lem:factorisation-into-power-series}, since $F_i(\veczero, y) = \prod_{i=1}^n (y - \alpha_i)$ is square-free, we have that $F_i(\vecz_i,y)$ factors 
\[
    F_i(\vecz_i,y) = \prod_{j=1}^n (y - A_{ij}(\vecz_i))
\]
where $A_{ij}(\vecz_i) \in \F\indsquare{\vecz_i}$ and $A_{ij}(\veczero) = \alpha_j$. Note that $\esym_r(A_{i1},\ldots, A_{in}) = z_{i,r} + \beta_r$. By \cref{thm:closure-powerseries-roots}, we have circuits $\tilde{A}_{ij}$ of size $\poly(s,n,d)$ and depth $\Delta+ O(1)$ such that 
\[
    \tilde{A}_{ij}(\vecz_i) = A_{ij}(\vecz_i) \bmod{\inangle{\vecz_i}^{d+1}}.
\]
From \cref{thm:complexity-of-multi-symmetric-polys}, degree of $Q$ is at most $d$. Thus,
\begin{align*}
P(\tilde{A}_{11}(\vecz_1), \ldots, \tilde{A}_{kn}(\vecz_k)) & =P(A_{11}(\vecz_1), \ldots, A_{kn}(\vecz_k)) \bmod{\inangle{\vecz}^{d+1}}\\
& = Q(\overline{\esym}(A_{11}, \ldots, A_{1n}), \dots, \overline{\esym}(A_{k1}, \ldots, A_{kn}))\\
& = Q(\vecz_1 + \vecbeta, \ldots, \vecz_k + \vecbeta)\\
\implies Q(\vecz_1 + \vecbeta, \ldots, \vecz_k + \vecbeta) & = \homog_{\leq d}P(\tilde{A}_{11}(\vecz_1), \ldots, \tilde{A}_{kn}(\vecz_k))
\end{align*}
By \cref{cor:interpolation-consequences}, the RHS is computable by a $\poly(n,d,s)$-size depth $\Delta + O(1)$ circuit, and by un-translating by $\vecbeta$ we also obtain a such a circuit for $Q(\vecz_1, \ldots, \vecz_k)$. 
\end{proof}

\section{Computing resultants and GCD}\label{sec:GCD}

\subsection{A constant-depth circuit for the resultant over all fields}

An almost immediate corollary of \cref{thm:complexity-of-multi-symmetric-polys} is a polynomial size constant-depth circuit for the resultant of two polynomials over any large enough field. 

\begin{corollary}[Computing resultants in constant depth]
    Let $\F$ be any large enough field (i.e., $|\F| \geq d^c$ for some absolute constant $c > 0$). 
    Given two degree $d$ univariate polynomials $f(y)$ and $g(y)$ that are monic in $y$ and provided as a list of coefficients, we can compute the resultant of the two polynomials, $\Res{y}{f}{g}$, via a $\poly(d)$-size depth $O(1)$ circuit over the coefficients of $f$ and $g$ over the field $\F$.  
\end{corollary}
\begin{proof}
    Let $f(y) = y^{d_1} - f_1 y^{d_1-1} + \cdots + (-1)^{d_1} f_0$ and $g(y) = y^{d_2} - g_1 y^{d_2-1} + \cdots + (-1)^{d_2} g_0$ and let $d = \max(d_1,d_2)$. If the multiset of roots (in $\overline{\F})$ of $f$ and $g$ are given by $\set{\sigma_1,\ldots, \sigma_{d_1}}$ and $\set{\tau_1,\ldots, \tau_{d_2}}$, then by \cref{fact:resultant-product-of-root-differences} we have
    \[
    \Res{y}{f}{g} = \prod_{i = 1}^{d_1} \prod_{j=1}^{d_2} (\sigma_i - \tau_j). 
    \]
    The RHS is clearly a bi-symmetric polynomial with respect to $\set{\sigma_1,\ldots, \sigma_{d_1}} \sqcup \set{\tau_1,\ldots, \tau_{d_2}}$ and is a $\poly(d)$ sized depth $2$ circuit over the $\sigma$'s and $\tau$'s. By \cref{thm:fundamental-theorem-of-multi-symm-poly}, there exists a unique $Q$ of degree at most $d_1 d_2$ such that 
    \[
        \Res{y}{f}{g} = Q(\overline{\esym}(\sigma_1,\ldots, \sigma_{d_1}),\overline{\esym}(\tau_1, \ldots, \tau_{d_2})). 
    \]
    Note that $\esym_r(\sigma_1,\ldots, \sigma_{d_1})$ is just $f_r$ and $\esym_r(\tau_1,\ldots, \tau_{d_2})$ is just $g_r$. Thus, by \cref{thm:complexity-of-multi-symmetric-polys}, $Q$ is also computable by a circuit of size $\poly(d)$ and depth $O(1)$ over the coefficients $f_1,\ldots, f_{d_1}, g_1,\ldots, g_{d_2}$ over the field $\F$. 
\end{proof}

If $f(\vecx, y), g(\vecx,y)$ are multivariate polynomials that are monic in $y$ and computable by size $s$ depth $\Delta$ circuits, then the coefficients of monomials in $y$ is also computable by size $\poly(s,n,d)$, depth $\Delta + O(1)$ circuits and thus $\Res{y}{f}{g}$ is also computable by a size $\poly(s,n,d)$ depth $\Delta+O(1)$ circuit. 

The resultant is the case of computing the product of $g$'s evaluations on the roots of $f$. A more general statement is the following lemma. 
    
\begin{lemma}
    \label{lem:esym-of-g-on-roots-of-f}
    Let $\F$ be any large enough field (i.e., $|\F| \geq (sd_1d_2n)^c$ for some absolute constant $c > 0$). 
    Suppose $f(\vecx,y), g(\vecx,y) \in \F[\vecx,y]$ are two n-variate polynomial that are monic in $y$, of degree $d_1$ and $d_2$ respectively (in $y$), that are given by size $s$, depth $\Delta$ circuits. Let $\set{\sigma_1,\ldots, \sigma_{d_1}}$ be the multi-set of $y$-roots of $f$ over $\overline{\F(\vecx)}$. Then, for any $1 \leq r \leq d_1$, we can obtain a circuit of size $\poly(s,n,d_1,d_2)$ and depth $\Delta + O(1)$ for 
    \[
    \esym_r(\set{g(\vecx, \sigma_1), \ldots, g(\vecx, \sigma_{d_1})}) \in \F[\vecx].
    \]
\end{lemma}
\begin{proof}
    Let $d = \max(d_1,d_2)$. Note that $\esym_r(y_1,\ldots, y_m) = \coeff{t^r}{(1 + ty_1)\cdots (1 + ty_m)}$. By \cref{cor:interpolation-consequences}, it is computable by a size $\poly(r,m)$ depth $3$ circuit over $\F$. 

    If $\set{\tau_1,\ldots, \tau_{d_2}}$ is the multi-set of $y$-roots of $g$, then 
    \[
    \esym_r(\set{g(\vecx, \sigma_1), \ldots, g(\vecx, \sigma_{d_1})}) = \esym_r\inparen{\setdef{\prod_{j=1}^{d_2}(\sigma_i - \tau_j)}{i\in [d_1]}}
    \]
    is clearly a bi-symmetric polynomial with respect to $\set{\sigma_1,\ldots, \sigma_{d_1}} \sqcup \set{\tau_1,\ldots, \tau_{d_2}}$. Thus, if $f(\vecx,y) = y^{d_1} + f_1 y^{d_1 - 1} + \cdots + f_{d_1}$ and $g(\vecx, y) = y^{d_2} + g_1 y^{d_2 - 1} + \cdots + g_0$, then each $f_i, g_i$ are the elementary symmetric polynomials of $\set{\sigma_1,\ldots, \sigma_{d_1}}$ and $\set{\tau_1,\ldots, \tau_{d_2}}$, and \cref{cor:interpolation-consequences} shows that we can obtain circuits of size $\poly(s,d)$ and depth $\Delta + O(1)$ for them. By \cref{thm:complexity-of-multi-symmetric-polys}, we can express the above as $Q(f_1, \ldots, f_{d_1}, g_1, \ldots, g_{d_2})$ where $f(\vecx,y) = y^{d_1} + f_1 y^{d_1 - 1} + \cdots + f_{d_1}$ and $g(\vecx, y) = y^{d_2} + g_1 y^{d_2 - 1} + \cdots + g_0$ with $Q$ computable by a size $\poly(s,n, d)$, depth $\Delta + O(1)$ circuit. 
\end{proof}

The above lemma can also be easily generalized to computing elementary symmetric polynomials of a \emph{rational function} $g(x)/h(x)$ on the roots of $f(x)$ in \cref{lem:esym-of-g/h-on-roots-of-f}. 

\subsection{Computing the GCD of two polynomials}

Unlike operations such as computing the resultant of two polynomials, the GCD of two polynomials $f(y), g(y)$ is not a continuous function of the coefficients and therefore cannot be expressed as a polynomial (or even a rational) function of the coefficients of $f$ and $g$. Andrews and Wigderson~\cite{AW24} show that it can be expressed as a \emph{piece-wise rational function} of the coefficients of $f$, $g$ where we have a sequence of rational functions computed by small algebraic circuits over the coefficients of $f$ and $g$ such that one of them computes $\gcd(f,g)$, and we know which of them via appropriate \emph{zero-tests}. 

\begin{definition}[Piece-wise rational circuit families]
    \label{defn:piece-wise-rational-fn} A \emph{piece-wise rational circuit} for a function $P: \F^n \rightarrow \F^m$ is a family of `computation circuits' $\set{A_1, \ldots, A_s: \F^n \rightarrow \F^m}$, $\set{B_1,\ldots, B_s:\F^n \rightarrow \F^m}$ and `test circuits' $\set{T_1,\ldots, T_s:\F^n \rightarrow \F}$ such that for every $\veca \in \F^n$ we have
    \[
    P(\veca) = \frac{A_r(\veca)}{B_r(\veca)}
    \]
    where the \emph{advice parameter} $r = \max \setdef{k \in \N}{T_k(\veca) \neq 0}$. The size of the family is said to be the sum of the sizes of $A_i$'s, $B_i$'s and $T_i$'s, and the depth of the family is maximum depth of any $A_i, B_i, T_i$ that constitute the family. 
\end{definition}

Situations where $P(\veca) \in \F[y]$ (such as computing the GCD of two polynomials given by a list of coefficients) can be equivalently viewed as a multi-output function that outputs the list of coefficients of the polynomial $P(\veca)$. For a more detailed discussion on this perspective and the definition of \emph{piece-wise rational $\mathrm{AC}^0$ circuits}, we refer the reader to \cite[Section 2 in the full version]{AW24}. 

\medskip

The following is the analogue of the main result of \cite{AW24} for computing the GCD of two polynomials over all sufficiently large fields expressed in terms of piece-wise rational circuit families.  

\begin{theorem}[Computing GCD of in terms of coefficients]
    \label{thm:gcd-computation}
    Let $\K$ be any polynomially large field (i.e., $|\K| \geq (d_1d_2)^c$ for some absolute constant $c > 0$). 
    There are explicit families of algebraic circuits $\set{P_{d_1,d_2,k}}$, $\set{Q_{d_1,d_2,k}}$, $\set{T_{d_1,d_2,k}}$ of size $\poly(d_1,d_2)$ and depth $O(1)$ such that, for every $f(y), g(y) \in \K[y]$ that are monic in $y$ with $\deg(f) = d_1, \deg(g) = d_2$, we have 
    \[
    \gcd(f(y),g(y)) = \frac{P_{d_1, d_2, r}(\overline{\mathrm{coeff}}(f), \overline{\mathrm{coeff}}(g))}{Q_{d_1, d_2, r}(\overline{\mathrm{coeff}}(f), \overline{\mathrm{coeff}}(g))}
    \]
    where $r = \max \setdef{k \leq \poly(d_1,d_2)}{T_{d_1,d_2,k}(\overline{\mathrm{coeff}}(f), \overline{\mathrm{coeff}}(g)) \neq 0}$. (Here will use $\overline{\mathrm{coeff}}(f)$ to refer to the vector of coefficients of $f$.)
\end{theorem}

In fact, the `advice' $r$ above is essentially just $\deg(\gcd(f,g))$. The key idea is to consider the polynomial $F(y,z) = f(y) + z \cdot g(y) = \gcd(f,g) \cdot F'(y,z)$ and somehow `filter out' the $F'(y,z)$. We now describe how to accomplish this. 

\medskip

Let $f(y) \in \K[y]$ be a monic polynomial that factorizes over $\overline{\K}$ as $(y - \sigma_1)^{e_1} \cdots (y - \sigma_m)^{e_m}$, with $e_i \geq 1$. For any polynomial $g(y) \in \K[y]$, we define $\Filter(f \mid g=0)$ and $\Filter(f \mid g \neq 0)$ as
\begin{align*}
    \Filter(f \mid g \neq 0) &= \prod_{\substack{i\in [m]\\g(\sigma_i) \neq 0}} (y - \sigma_i)^{e_i},\\
    \Filter(f \mid g=0) &= \prod_{\substack{i\in [m]\\g(\sigma_i) = 0}} (y - \sigma_i)^{e_i} =  \frac{f}{\Filter(f \mid g \neq 0)}.
\end{align*}

\begin{lemma}[Computing $\Filter$ from coefficients]
    \label{lem:computing-filter}
    Let $\K$ be any large enough field (i.e., $|\K| \geq (d_1d_2)^c$ for some absolute constant $c > 0$). 
    There are explicit families of circuits $\set{A_{d_1,d_2,k}}$, $\set{B_{d_1,d_2,k}}$, $\set{T_{d_1,d_2,k}}$ of size $\poly(d_1,d_2)$ and depth $O(1)$ such that, for every pair of monic polynomial $f,g \in \K[y]$ with $\deg f = d_1$ and $\deg g = d_2$, we have 
    \[
        \Filter(f \mid g \neq 0) = \frac{A_{d_1, d_2, r}(\overline{\mathrm{coeff}}(f), \overline{\mathrm{coeff}}(g))}{B_{d_1, d_2, r}(\overline{\mathrm{coeff}}(f), \overline{\mathrm{coeff}}(g))}
    \]
    where the advice parameter $r = \max \setdef{k \leq \poly(d_1,d_2)}{T_{d_1,d_2,k}(\overline{\mathrm{coeff}}(f), \overline{\mathrm{coeff}}(g))\neq 0}$. 
\end{lemma}
\begin{proof}
    Let $f(y) = \prod_{i=1}^m (y - \sigma_i)^{e_i}$ with $\sigma_i \in \overline{\K}$ and $e_i \geq 1$, and let $\Sigma$ be the multiset of roots of $f(y)$ in $\overline{\K}$. If $r = \deg(\Filter(f \mid g\neq 0))$, then
    \begin{align*}
        \esym_{r} \inparen{\setdef{g(\sigma)}{\sigma \in \Sigma}} & = \prod_{i \in [m]\;:\;g(\sigma_i) \neq 0} g(\sigma_i)^{e_i} \neq 0,\\
        \esym_{k} \inparen{\setdef{g(\sigma)}{\sigma \in \Sigma}} &  = 0 \quad \text{for all $k > r$}
    \end{align*}
    as any subset of size greater than $r$ among $\Sigma$ must necessarily contain a $\sigma$ such that $g(\sigma) = 0$. 
    By \cref{lem:esym-of-g-on-roots-of-f}, the above is expressible as circuit of size $\poly(d_1, d_2)$ and depth $O(1)$ over the coefficients of $f$ and $g$ and that gives us the claimed family $\set{T_{d_1,d_2,k}}$.

    \medskip
    To obtain the filter as a ratio of algebraic circuits, define the polynomial $h(x,y) = (y - x) \cdot g(x)$ for a fresh variable $x$. Note that, for any $\sigma \in \overline{\K}$, we have $h(\sigma, y) = 0$ if and only if $g(\sigma) = 0$. Therefore, 
    \begin{align*}
    \esym_{r} \inparen{\setdef{h(\sigma, y)}{\sigma \in \Sigma}} & = \prod_{i \in [m]\;:\;g(\sigma_i) \neq 0} \inparen{(y - \sigma_i) \cdot g(\sigma_i)}^{e_i}\\
    & = \inparen{\prod_{i\in [m]\;:\;g(\sigma_i) \neq 0} (y - \sigma_i)^{e_i}} \cdot \esym_{r} \inparen{\setdef{g(\sigma)}{\sigma \in \Sigma}}\\
    \therefore \Filter(f\mid g \neq 0) & = \frac{\esym_{r} \inparen{\setdef{h(\sigma, y)}{\sigma \in \Sigma}}}{\esym_{r} \inparen{\setdef{g(\sigma)}{\sigma \in \Sigma}}}
    \end{align*}
    Once again by \cref{lem:esym-of-g-on-roots-of-f}, both the numerator and denominator on the RHS can be expressed as circuits of size $\poly(d_1, d_2)$ and depth $O(1)$ over the coefficients of $f$ and $g$ and that gives us the circuit families $\set{A_{d_1,d_2,k}}$ and $\set{B_{d_1,d_2,k}}$. 
\end{proof}

We are now ready to prove \cref{thm:gcd-computation}

\begin{proof}[Proof of \cref{thm:gcd-computation}]
    Without loss of generality, let us assume that $d_1 := \deg f > \deg g =: d_2$ (by working with $f, g-f$ if their degrees were the same). Let $\set{\sigma_1,\ldots, \sigma_m} \in \overline{\K}$ be the union of roots of $f$ and $g$, and let 
    \begin{align*}
    f(y) & = (y - \sigma_1)^{e_1} \cdots (y - \sigma_m)^{e_m},\\
    g(y) & = (y - \sigma_1)^{e_1'} \cdots (y - \sigma_m)^{e_m'}\\
    \implies \gcd(f,g) & = (y - \sigma_1)^{e_1^*} \cdots (y - \sigma_m)^{e_m^*} \quad \text{where }e_i^* = \min(e_i, e_i').
    \end{align*}

    Consider the polynomial $F(y,z) = f(y) + z \cdot g(y)$ for a fresh variable $z$, which is monic in $y$ since $\deg f > \deg g$. The key observation is that $\gcd(f,g) = \Filter(F(y,z) \mid g = 0)$ and the theorem essentially follows from \cref{lem:computing-filter} and this observation. We formalize this below. 

    \medskip
    \noindent
    Since $F(y,z)$ is clearly divisible by $\gcd(f,g)$, we can write $F(y,z)$ as 
    \[
    F(y,z) = \gcd(f,g) \cdot F'(y,z)
    \]
    for some $F'(y,z) \in \K[y,z]$. Note that, for every $\sigma_i$, we have $F'(\sigma_i,z) \neq 0$ (for otherwise $(y - \sigma_i)^{e_i^* + 1}$ divides both $f$ and $g$). Therefore, 
    \begin{align*}
    \Filter(F(y,z) \mid g \neq 0) &= F'(y,z)\\
    \implies \Filter(F(y,z) \mid g = 0) &= \frac{F(y,z)}{\Filter(F(y,z) \mid g \neq 0)} = \gcd(f,g).
    \end{align*}
    By \cref{lem:computing-filter}, we have circuit families $\set{A_{d_1,d_2,k}}$, $\set{B_{d_1,d_2,k}}$, $\set{T_{d_1,d_2,k}}$ such that 
    \begin{align*}
    \Filter(F(y,z) \mid g \neq 0) & = \frac{A_{d_1, d_2, r}(\overline{\mathrm{coeff}_y}(F), \overline{\mathrm{coeff}}(g))}{B_{d_1, d_2, r}(\overline{\mathrm{coeff}_y}(F), \overline{\mathrm{coeff}}(g))}\\
    \text{where } r &= \max \setdef{k \leq \poly(d_1,d_2)}{T_{d_1,d_2,k}(\overline{\mathrm{coeff}_y}(F), \overline{\mathrm{coeff}}(g)) \neq 0}\\
    \implies \gcd(f,g) &= \Filter(F(y,z) \mid g = 0) \\
    &= \frac{F(y,z) \cdot B_{d_1, d_2, r}(\overline{\mathrm{coeff}_y}(F), \overline{\mathrm{coeff}}(g))}{A_{d_1, d_2, r}(\overline{\mathrm{coeff}_y}(F), \overline{\mathrm{coeff}}(g))}.
    \end{align*}
    Since each coefficient of $F$ can be expressed as a polynomial function of the coefficients of $f$ and $g$, the above circuit families $\set{A_{d_1,d_2,k}}$, $\set{B_{d_1,d_2,k}}$, $\set{T_{d_1,d_2,k}}$ can also be expressed as $\poly(d_1,d_2)$ size depth $O(1)$ circuit families over the coefficients of $f$ and $g$. It can also be checked that the advice parameter $r = d_1 - \deg(\gcd(f,g))$. 

    \medskip

    As stated, although $\gcd(f(y),g(y)) \in \K[y]$ and is thus independent of the variable $z$, the RHS of the above equations consist of numerators and denominators that depend on $z$. To eliminate $z$ from the RHS, we can remove any powers of $z$ dividing the numerator and denominator and then set $z$ to zero. This can be achieved by expanding the test circuits to also find the smallest $i$ such that $\coeff{z^i}{A_{d_1, d_2, r}(\overline{\mathrm{coeff}_y}(F), \overline{\mathrm{coeff}}(g))} \neq 0$ and set the numerator and denominator circuits as
    \begin{align*}
    P_{d_1,d_2,(r,i)}(\overline{\mathrm{coeff}}(f), \overline{\mathrm{coeff}}(g)) &:= \coeff{z^i}{F(y,z) \cdot B_{d_1, d_2, r}(\overline{\mathrm{coeff}_y}(F), \overline{\mathrm{coeff}}(g))},\\
    Q_{d_1,d_2,(r,i)}(\overline{\mathrm{coeff}}(f), \overline{\mathrm{coeff}}(g)) &:= \coeff{z^i}{A_{d_1, d_2, r}(\overline{\mathrm{coeff}_y}(F), \overline{\mathrm{coeff}}(g))}
    \end{align*}
    respectively. Thus, overall, we obtain size $\poly(d_1,d_2)$ and depth $O(1)$ circuits over the coefficients of $f$ and $g$ as claimed by the theorem. 
\end{proof}

\section{Improved closure results for fields of small characteristic}\label{sec:finite-fields-closure}
\label{sec:closure-small-char}

In this section, we use the \autoref{thm:intro-complexity-of-symmetric-polys} to prove the following closure result for constant-depth circuits over fields of characteristic $p$. This statement is similar to a result in \cite{BKRRSS25} with one key difference -- the circuit for $g(\vecx)^{p^\ell}$ in \autoref{thm:closure-factors-small-char} is over the underlying field itself and not over the algebraic closure (or a high degree extension) of the base field $\F_q$, provided $q$ is polynomially large. This extends the results of Kaltofen~\cite{K89} for fields of characteristic $p$ for natural subclasses of algebraic circuits such as constant-depth circuits, formulas, branching programs etc. The key technical ingredient in the proof again is the use of \autoref{thm:intro-complexity-of-symmetric-polys}.

\begin{theorem}[Complexity of factors]
    \label{thm:closure-factors-small-char}
    Let $\F_q$ be a field of positive characteristic $p$ that is polynomially large (i.e., there is an absolute constant $c$ such that $\abs{\F} \geq (snd)^c$). Let $P(\vecx) \in \F_q[\vecx]$ be a polynomial on $n$ variables of degree $d$ computed by a circuit $C$ of size $s$ and depth $\Delta$. Further, let $g(\vecx)$ be a factor of $P(\vecx)$ with multiplicity $p^\ell\cdot e$ where $\gcd(p,e) = 1$. Then, $g(\vecx)^{p^\ell}$ is computable by a circuit of size $\poly(s,d,n)$ and depth $\Delta + O(1)$ over $\F_q$. 
\end{theorem}

We start with a few preliminaries necessary for the proof of the theorem. We shall interpret $P(\vecx)$ as an element of $\F[\vecx, y]$, and apply a map $\Psi: x_i \mapsto tx_i + a_i y + b_i$ on $P(\vecx)$ for randomly chosen $a_i$s and $b_i$s, to ensure certain properties\footnote{We expand on this in the beginning of \autoref{sec:proof-of-charp-closure-general-factors}. For more details, please refer to the preliminaries of \cite{BKRRSS25}.}. Thus, without loss of generality, we can work with a bivariate polynomial $P(t,y)$ over a field that contains the other variables. Some preliminaries will also be stated for bivariate polynomials.

\subsection{Hasse derivatives}

We shall work with the notion of Hasse derivative, which is the standard alternative to partial derivatives in the small characteristic setting. We state the definition and the product rule for Hasse derivatives. For more details, we recommend the reader to refer to \cite[Appendix C]{forbes-thesis-2014}.

\begin{definition}[Hasse derivatives]
    \label{defn:hasse-derivative}
    The \emph{Hasse Derivative of order $i$} of $F(t, y) \in \F[t, y]$ with respect to $y$, denoted as $\hasse{i}{y}(F)$, is defined as the coefficient of $z^i$ in the polynomial $F(t, y+z)$.
\end{definition}

\begin{lemma}[Product rule for Hasse derivatives] \label{lem:product-rule-hasse}
    Let $G(t, y), H(t,y) \in \F[t, y]$ be bivariate polynomials and let $k \geq 0$. Then,
    \[
    \hasse{k}{y}(GH) = \sum_{i+j=k} \hasse{i}{y}(G) \cdot \hasse{j}{y}(H)
    \]
\end{lemma}

\subsection{Furstenberg's theorem over small characteristic fields}

The following version of Furstenberg's theorem over small characteristic is very similar to \cref{thm:furstenberg}, with some key differences. The theorem expresses an appropriate power of a power series root of a polynomial as a diagonal of a rational expression involving the polynomial and its derivatives. 

\begin{theorem}[Furstenberg's theorem over small characteristic fields (Theorem A.3 in \cite{BKRRSS25})] \label{thm:furstenberg-small-characteristic}
    Let $\F$ be a field of characteristic $p$. Let $P(t,y) \in \F \indsquare{t,y}$ be a power series and $\varphi(t){\in \F\indsquare{t}}$ be a power series satisfying 
    \[
    P(t, y) = (y - \phi(t))^{p^\ell e} \cdot Q(t,y)
    \]
    for some $\ell \geq 0, e \geq 1$ such that $\gcd(p, e) = 1$. If $\phi(0) = 0$ and $Q(0,0) \neq 0$, then
    \begin{equation}
        \label{eqn: furstenberg-expression-small-char}
        \varphi^{p^\ell} = \diag\inparen{\frac{y^{2p^\ell} \cdot \hasse{p^\ell}{y}(P)(ty,y)}{e \cdot P(ty,y)}}
    \end{equation}
\end{theorem}
We can further simplify the expression in \autoref{thm:furstenberg-small-characteristic} to get a version of \autoref{cor:flajolet-soria-formula-for-roots} over small characteristic fields. 

\begin{corollary}[Corollary A.5 in \cite{BKRRSS25}]
    \label{cor:charp-flajolet-soria-formula-for-roots}
    Let $P(t,y), Q(t,y) \in \F\indsquare{t,y}$ and $\phi(t) \in \F\indsquare{t}$ satisfy
    \[
    P(t, y) = (y - \phi(t))^{p^\ell e} \cdot Q(t,y)
    \]
    with $\gcd(p,e) = 1$, $\phi(0) = 0$ and $Q(0,0) = \alpha \neq 0$. Then, 
    \[
    \phi(t)^{p^\ell} = \sum_{m\geq 0} \coeff{y^{p^{\ell}(e(m+1)-2)}}{\frac{\hasse{p^\ell}{y}(P)(t,y)}{e \cdot \alpha^{m+1}}{\inparen{\alpha y^{p^\ell\cdot e}-P(t,y)}^m}}.
    \]
    Moreover,
    \[
    \homog_{\leq d}[\phi(t)^{p^\ell}] = \homog_{\leq d}\insquare{\sum_{m\geq 0}^{2e(d+p^\ell) } \coeff{y^{p^{\ell}(e(m+1)-2)}}{\frac{\hasse{p^\ell}{y}(P)(t,y)}{e \cdot \alpha^{m+1}}{\inparen{\alpha y^{p^\ell\cdot e}-P(t,y)}^m}}}.
    \]
\end{corollary}

The following claim would be helpful to express $\alpha = Q(0,0)$ as an appropriate Hasse derivative of $P$. 

\begin{claim} \label{claim:small-char-Q(0,0)=hasse(P)(0,0)}
    Suppose $P(t, y) = (y - \phi(t))^{p^\ell e} \cdot Q(t,y) \in \F\indsquare{t,y}$ with $\gcd(p,e) = 1$, $\phi(0) = 0$ and $Q(0,0) = \alpha \neq 0$. Then, 
    \[
    \alpha = \hasse{p^\ell \cdot e}{y}P(0,0)
    \]
\end{claim}
\begin{proof}
    Let $F(t,y) = (y - \phi(t))^{p^\ell \cdot e}$. Then, by \cref{lem:product-rule-hasse}, 
    \begin{align*}
    \hasse{p^\ell \cdot e}{y}P(t,y) &= \sum_{i + j = p^\ell \cdot e} \hasse{i}{y} F(t,y) \cdot \hasse{j}{y}Q(t,y)\\
    \implies \hasse{p^\ell \cdot e}{y}P(0,0) &= \sum_{i + j = p^\ell \cdot e} \hasse{i}{y} F(0,0) \cdot \hasse{j}{y}Q(0,0)
    \end{align*}
    Note that $\hasse{i}{y}F = \coeff{z^i}{(y^{p^\ell} + z^{p^\ell} - \phi^{p^\ell})^e}$ and is a nonzero polynomial only for $i$ of the form $p^\ell \cdot e'$ for some $0 \leq e' \leq e$, and $\hasse{p^\ell e}{y}F = 1$. Furthermore, for every $e' < e$, we have $\hasse{p^\ell e'}{y}F$ is divisible by $(y - \phi)^{p^\ell}$ and hence $\hasse{p^\ell e'}{y}F(0,0) = 0$ for all $e' < e$. Therefore,
    \begin{align*}
    \hasse{p^\ell \cdot e}{y}P(0,0) &= \sum_{i + j = p^\ell \cdot e} \inparen{\hasse{i}{y} F(0,0)} \cdot \inparen{\hasse{j}{y}Q(0,0)}\\
    & = \inparen{\hasse{p^\ell e}{y}F(0,0)} \cdot \inparen{\hasse{0}{y}Q(0,0)} = Q(0,0) = \alpha.\qedhere
    \end{align*}
\end{proof}

\subsection{Complexity of power series roots over small characteristic fields}

Using \cref{thm:furstenberg-small-characteristic} and \autoref{cor:charp-flajolet-soria-formula-for-roots}, we get the following analogue of \cref{thm:closure-powerseries-roots} over arbitrary fields of small characteristic. 

\begin{theorem}[Power series roots with multiplicity over small characteristic]
    \label{thm:closure-powerseries-roots-small-char}
    Let $\F$ be a field of positive characteristic $p$. 
    Suppose $P(\vecx,y) \in \F[\vecx,y]$ is a polynomial computed by a circuit $C$, and $\varphi(\vecx) {\in \F\indsquare{\vecx}}$ is a power series satisfying $P(\vecx,y) = (y-\varphi(\vecx))^{p^\ell e} \cdot Q(\vecx,y)$ where $\varphi(\veczero) = 0$, $\gcd(p,e) = 1$ and $Q(\veczero,0)\neq 0$. Then, for any $d \in \N$, there is a circuit $C'$ over $\F$ computing $\homog_{\leq d}\insquare{\varphi^{p^\ell}}$ such that
    \[\size(C') \leq \poly(d,\size(C))\]
    \[\depth(C') \leq \depth(C) + O(1)\]
\end{theorem}
\begin{proof}
    The theorem follows from \cref{cor:charp-flajolet-soria-formula-for-roots} that asserts that 
    \[
    \homog_{\leq d}\insquare{\varphi^{p^\ell}} = \homog_{\leq d}\insquare{\sum_{m\geq 0}^{2e(d+p^\ell) } \coeff{y^{p^{\ell}(e(m+1)-2)}}{\frac{\hasse{p^\ell}{y}(P)(t,y)}{e \cdot \alpha^{m+1}}{\inparen{\alpha y^{p^\ell\cdot e}-P(t,y)}^m}}},
    \]
    and \cref{cor:interpolation-consequences} shows that the RHS is computable by a size $\poly(s,n,d)$ and depth $\Delta + O(1)$ circuit. 
\end{proof}

\subsection{Proof of \autoref{thm:closure-factors-small-char}} \label{sec:proof-of-charp-closure-general-factors}

We will need the following generalization of \cref{lem:esym-of-g-on-roots-of-f}, whose proof we will defer to the end of this section. 

\begin{lemma}
    \label{lem:esym-of-g/h-on-roots-of-f}
    Let $\F$ be any large enough field (i.e., $|\K| \geq (sdn)^c$ for some absolute constant $c > 0$). Let $f(\vecx, y), g(\vecx, y), h(\vecx, y) \in \F[\vecx, y]$ be polynomials that are monic in $y$ of degrees $d_f$, $d_g$ and $d_h$ respectively, with $\gcd(f, h) = 1$, that are given by size $s$, depth $\Delta$ circuits; let $d = \max(d_f, d_g, d_h)$. If $\set{\sigma_1,\dots, \sigma_{d_f}}$ be the multi-set of $y$-roots of $f$ over $\overline{\F(\vecx)}$, then for any $1 \leq r \leq d_f$, we can get circuits $A(\vecx), B(\vecx)$ of size $\poly(s,n,d)$ and depth $\Delta + O(1)$ such that 
    \[
    \esym_r\inparen{\setdef{\frac{g(\vecx, \sigma_i)}{h(\vecx, \sigma_i)}}{i \in [d_f]}} = \frac{A(\vecx)}{B(\vecx)}.
    \]
\end{lemma}

\noindent
We are now ready to complete the proof of \autoref{thm:closure-factors-small-char}. 

\begin{proof}[Proof of \autoref{thm:closure-factors-small-char}]
    By interpreting $P(\vecx)$ as an element of $\F[\vecx, y]$, let $\tilde{P}(t,y) = \Psi(P(\vecx,y)) \in \F[\vecx][t,y]$ for a map $\Psi: x_i \mapsto tx_i + a_i y + b_i$ where the $a_i$s and $b_i$s are chosen at random. This will ensure that $\tilde{P}(t,y)$ is monic in $y$, the square-free part of $\tilde{P}(t,y)$ is also square-free at $t=0$, and that $\tilde{P}(t,y) \in \F[\vecx][t,y]$ is also computable by a size $\poly(s,d,n)$, depth $\Delta + O(1)$ circuit over $\F$. From Gauss' Lemma, we know that all the factors of $\tilde{P}$ can now be assumed to be monic in $y$ without loss of generality, and viewing $\tilde{P}$ as a bivariate in $t$ and $y$ with coefficients in $\F(\vecx)$ maintains the factorization pattern of the polynomial. Furthermore, the map can be inverted using the transformation $x_i \mapsto x_i - a_i y - b_i$ and $t \mapsto 1$ and this gives us a way to argue about factors of $P$ using factors of $\tilde{P}$. Thus\footnote{For more details, we encourage the reader to refer to the preliminaries in \cite{BKRRSS25}.}, it suffices to show that for an arbitrary factor $g(t,y) \in \F[\vecx][t,y]$ of $\tilde{P}(t,y)$, if the multiplicity of $g(t,y)$ is $p^\ell \cdot e$ with $\gcd(e,p) = 1$, then $g(t,y)^{p^\ell}$ is computable by a $\poly(s,d,n)$ size, depth $\Delta + O(1)$ circuit over the field $\F$.
 
    First we note that the existence of a power series factorization of $\tilde{P}(t,y)$ follows \emph{almost} immediately from \autoref{lem:factorisation-into-power-series}. We cannot apply \autoref{lem:factorisation-into-power-series} directly because the preprocessing map $\Psi$ only ensures that the square-free part of $\tilde{P}(t,y)$ remains square-free when $t$ is set to zero. Applying \autoref{lem:factorisation-into-power-series} on the square-free part of $\tilde{P}(t,y)$, denoted by $\hat{P}(t,y)$,  gives us the factorization $\hat{P}(t,y) = \prod_{i \in [m]}(y-\varphi_i(t))$, where each $\varphi_i(t)$ has a distinct constant term $\varphi_i(0) \in \overline{\F}$. 

    \medskip

    Suppose $g(t,y)$ is a factor with multiplicity $p^\ell \cdot e$. For each power series $y$-root $\phi(t)$ of $g(t,y)$, we can write $\tilde{P}(t,y) = (y-\varphi(t))^{p^\ell \cdot e}Q(t,y)$ for some $Q(t,y)$ that satisfies $Q(0,\varphi(0)) \neq 0$. Translating $y$ to $y+\varphi(0)$, we can use \autoref{cor:charp-flajolet-soria-formula-for-roots} on $\tilde{P}(t,y+\varphi(0))$ to show that $\homog_{\leq d}[\phi(t)^{p^\ell}]$ is equal to
    \[
        \homog_{\leq d}\insquare{\varphi(0)^{p^\ell} + \sum_{m\geq 0}^{2e(d+p^\ell)} \coeff{y^{p^{\ell}(e(m+1)-2)}}{\frac{\hasse{p^\ell}{y}(\tilde{P})(t,y+\varphi(0))}{e \cdot \alpha^{m+1}}{\inparen{\alpha (y+\varphi(0))^{p^\ell\cdot e}-\tilde{P}(t,y+\varphi(0))}^m}}}
    \]
    where $\alpha = Q(0,\varphi(0)) \neq 0$, which is also equal to $\hasse{p^\ell \cdot e}{y}\tilde{P}(0,\varphi(0))$ by \autoref{claim:small-char-Q(0,0)=hasse(P)(0,0)}. Thus, we define the rational function $R(z)$ as
    \[
        z^{p^\ell} + \sum_{m\geq 0}^{2e(d+p^\ell)} \coeff{y^{p^{\ell}(e(m+1)-2)}}{\frac{\hasse{p^\ell}{y}(\tilde{P})(t,y+z)}{e \cdot \hasse{p^\ell\cdot e}{y}(\tilde{P})(0,z)^{m+1}}{\inparen{\hasse{p^\ell\cdot e}{y}(\tilde{P})(0,z)\cdot (y+z)^{p^\ell\cdot e}-\tilde{P}(t,y+z)}^m}}.    
    \]
    If $\hat{P}(0,y) = \prod_{i=1}^m (y-\varphi_i(0))$, and $g(0,y)^{p^\ell} = \prod_{i \in S}(y-\varphi_i(0))^{p^\ell}$ for some subset $S \subset [m]$, then the power series roots of $g(t,y)$ are precisely $\phi_j(t)$ for $j \in S$ and hence
    \begin{align*}
    g(t, y)^{p^\ell} & = \prod_{j \in S} (y^{p^\ell} - \phi_j(t)^{p^\ell}) = \prod_{j \in S} (y^{p^\ell} - {R}(\phi_j(0))) \bmod t^{d+1} \\
    \implies g(t,y)^{p^\ell} & = \homog_{\leq d}\insquare{g'(t,y)}\\
    \text{where} \quad g'(t,y) & = \prod_{j \in S}(y^{p^\ell} - R(\phi_j(0)))\\
     & = \sum_{k=0}^{|S|} y^{p^\ell (|S| - k)} \cdot (-1)^{|S| - k} \cdot \esym_k(\setdef{R(\phi_j(0))}{j \in S})
    \end{align*}
    Thus, it suffices to show that $\esym_k(\setdef{R(\phi_j(0))}{j \in S})$ for all $1 \leq k \leq d$ can be computed by a $\poly(s,n,d)$ size depth $\Delta + O(1)$ circuit as that would immediately yield similar size and depth bounds for $g(t,y)^{p^\ell}$ by \cref{cor:interpolation-consequences}. 

    Note that the rational function $R(z)$ can be easily expressed as $\frac{R_{\text{num}}(z)}{R_{\text{denom}}(z)}$ where $R_{\text{num}}(z)$ and $R_{\text{denom}}(z)$ are both computable by a $\poly(s,d,n)$ sized depth $\Delta + O(1)$ circuits since $\tilde{P}(t,y)$ is given by a $\poly(s,d,n)$ sized depth $\Delta + O(1)$ circuit (using \cref{lem:interpolation} and \cref{cor:interpolation-consequences}). Furthermore, $R_{\text{denom}}(z) = (\hasse{p^\ell\cdot e}{y}(\tilde{P})(0,z))^{2e(d+p^\ell) + 1}$, so it remains nonzero element of $\overline{\F_q}$ when $z$ is set to any root of $\tilde{P}(0,y)$ (\autoref{claim:small-char-Q(0,0)=hasse(P)(0,0)}). 
    
    Note that $\esym_k(\setdef{\phi_j(0)}{j \in S})$ is in $\F$ since they are the (signed) coefficients of $g(0, y)$. Thus, by \cref{lem:esym-of-g/h-on-roots-of-f}, we can express $\esym_k(\setdef{R(\phi_j(0))}{j \in S})$ as a ratio of two $\poly(s,d,n)$ size depth $\Delta + O(1)$ circuits over $\F$. By eliminating the division gate using the standard technique of Strassen (which works in constant depth), we obtain a $\poly(s,d,n)$ size depth $\Delta + O(1)$ circuits over $\F$ for $\esym_k(\setdef{R(\phi_j(0))}{j \in S})$ for each $1 \leq k \leq d$, and thus a similar bound for $g(t,y)^{p^\ell}$. 
\end{proof}

\begin{proof}[Proof of \cref{lem:esym-of-g/h-on-roots-of-f}]
    Note that 
    \begin{equation}
    \esym_r\inparen{\frac{u_1}{v_1},\ldots, \frac{u_{d_f}}{v_{d_f}}} = \coeff{t^r}{(v_1 + tu_1)\cdots (v_{d_f} + tu_{d_f})} \cdot \frac{1}{v_1\dots v_{d_f}}. \label{eqn:eval-g/h-on-roots-of-f}
    \end{equation}
    Let $\set{\tau_1,\dots, \tau_{d_g}} , \set{\rho_1, \dots, \rho_{d_h}} \subset \overline{\F(\vecx)}$ be the multi-set of $y$-roots of $g$ and $h$ respectively. Thus, setting $u_i = g(\vecx, \sigma_i)$ and $v_i = h(\vecx,\sigma_i)$ in \eqref{eqn:eval-g/h-on-roots-of-f}, we see that 
    \begin{align*}
    \Gamma({\bm \sigma}, {\bm \tau}, {\bm \rho}) &= \coeff{t^r}{(h(\vecx, \sigma_1) + t g(\vecx, \sigma_1))\cdots (h(\vecx, \sigma_{d_f}) + t g(\vecx, \sigma_{d_f}))}\\
    & = \coeff{t^r}{\inparen{\prod_{j=1}^{d_h}(\sigma_1 - \rho_j) + t \prod_{k=1}^{d_g}(\sigma_1 - \tau_k)}\cdots \inparen{\prod_{j=1}^{d_h}(\sigma_{d_f} - \rho_j) + t \prod_{k=1}^{d_g}(\sigma_{d_f} - \tau_k)}}
    \end{align*}
    is multi-symmetric with respect to $\set{\sigma_1,\dots, \sigma_{d_f}} \sqcup \set{\tau_1,\dots, \tau_{d_g}} \sqcup \set{\rho_1, \dots, \rho_{d_h}}$ that is computable by a $\poly(s,d)$ size depth $\Delta + O(1)$ circuit (over ${\bm \sigma}, {\bm \tau}, {\bm \rho}$). Thus, by \cref{thm:complexity-of-multi-symmetric-polys}, we have a circuit of size $\poly(s,d)$ and depth $\Delta + O(1)$ over the elementary symmetric polynomials of ${\bm \sigma}, {\bm \tau}, {\bm \rho}$, which are the coefficients of $f$, $g$ and $h$, and they are also computable by size $\poly(s,d)$ depth $\Delta + O(1)$ circuits (by \cref{cor:interpolation-consequences}). That gives us the circuit $A(\vecx)$. 

    The circuit $B(\vecx)$ is to compute $\prod_{i=1}^{d_f} h(\vecx, \sigma_i) = \esym_{d_f}(h(\vecx, \sigma_1), \ldots, h(\vecx, \sigma_{d_f}))$ and \cref{lem:esym-of-g-on-roots-of-f} shows that $B(\vecx)$ is also computable by a size $\poly(s,d)$ depth $\Delta + O(1)$ circuit. Since $\gcd(f,h) = 1$, they do not share any $y$-roots and hence the division $A(\vecx)/B(\vecx)$ is well-defined.
\end{proof}

\ifblind
\else
\paragraph*{Acknowledgements:} The discussions leading to this work started when a subset of the authors were at the workshop on Algebraic and Analytic Methods in Computational Complexity (Dagstuhl Seminar 24381) at Schloss Dagstuhl, and continued when they met again during the HDX \& Codes workshop at ICTS-TIFR in Bengaluru.  We are thankful to the organizers of these workshops and to the staff at these centers for the wonderful collaborative atmosphere that facilitated these discussions. 

Varun Ramanathan is grateful to Srikanth Srinivasan and Amik Raj Behera at the University of Copenhagen for helpful discussions on the complexity of symmetric polynomials.
\fi

{\let\thefootnote\relax
\footnotetext{\textcolor{\gitinfonotecolour}{\gitinfonote \easteregg}
}}
\bibliographystyle{customurlbst/alphaurlpp}
\bibliography{crossref,references}

\end{document}